\documentclass{article}
\usepackage{spconf}

\usepackage{algorithm, algorithmic} 
\usepackage{float}  
\usepackage{lipsum}

\usepackage{cite}
\usepackage{diagbox}

\usepackage{amsmath,amssymb}  
\usepackage{mathrsfs}
\usepackage{graphicx}   
\usepackage{amsthm} 
\usepackage{setspace}
\usepackage{subfigure}
\usepackage{enumerate}

\usepackage{verbatim}  
\usepackage{array} 

\usepackage{multirow}
\usepackage{url}  

\def\R{{\mathbb{R}}}

\def\U{{\mathbf{U}}}
\def\A{{\mathbf{A}}}

\def\X{{\mathbf{X}}}

\def\F{{\mathscr{F}}}

\def\G{{\mathcal{G}}}

\def\E{{\mathcal{E}}}
\def\V{{\mathcal{V}}}
\def\S{{\mathcal{S}}}
\def\I{{\mathcal{I}}}
\def\B{{\mathcal{B}}}
\def\W{{\mathcal{W}}}

\newtheorem{theorem}{Theorem}
\newtheorem{definition}{Definition}
\newtheorem{lemma}{Lemma}

\newtheorem{proposition}{Proposition}

\title{SAMPLING OF CORRELATED BANDLIMITED CONTINUOUS SIGNALS BY JOINT TIME-VERTEX GRAPH FOURIER TRANSFORM}
\name{Zhongyi Ni$^\ast$, Feng Ji$^\dagger$, Hang Sheng$^\ast$, Hui Feng$^\ast$, Bo Hu$^\ast$}
\address{$^\ast$ School of Information Science and Engineering, Fudan University, Shanghai 200433, China\\
$^\dagger$ School of Electrical and Electronic Engineering, Nanyang Technological University, Singapore\\
$^\ast\lbrace$20300720007, 20110720036, hfeng, bohu$\rbrace$@fudan.edu.cn, $^\dagger$jifeng@ntu.edu.sg}

\begin{document}

\maketitle
\begin{abstract}

When sampling multiple signals, the correlation between the signals can be exploited to reduce the overall number of samples. In this paper, we study the sampling theory of multiple correlated signals, using correlation to sample them at the lowest sampling rate. Based on the correlation between signal sources, we model multiple continuous-time signals as continuous time-vertex graph signals. The graph signals are projected onto orthogonal bases to remove spatial correlation and reduce dimensions by graph Fourier transform. When the bandwidths of the original signals and the reduced dimension signals are given, we prove the minimum sampling rate required for recovery of the original signals, and propose a feasible sampling scheme.


\end{abstract}
%

\section{Introduction}

Sampling of continuous-time signals not only plays a crucial role in digital signal processing systems but also provides significant savings in the cost of signal storage and processing. The theory of perfect sampling and recovery for a bandlimited signal has been extensively studied\cite{shannon,landau}. In many scenarios, we need to observe multiple signal sources at the same time, such as biological research, sensor networks, and social networks. In these cases, the observed continuous-time signals may be correlated with each other, which introduces additional information. That is, the observations from sources at a moment constitute high-dimensional and redundant data. It would be wasteful to still sample each signal according to its bandwidth.

For the sampling theory of multiple signals, \cite{multisignal1974} studied the problem of the exact recovery of multidimensional signals from projections. \cite{corsignal,corsignal2} relate the bandwidth and correlation to the sampling rate required for signal reconstruction. In general, we hope to further reduce the data dimension by correlation. Principal Component Analysis (PCA) is one of the most common ways to reduce the dimension of data by projecting high-dimensional data into an orthogonal space\cite{PCA}. This dimension reduction process is consistent with graph signal processing which adopts the adjacency matrix as its fundamental framework\cite{GFTW}. The covariance matrix is the adjacency matrix, which is used to characterize the correlation between signal sources. Eigenvectors of the adjacency matrix are defined as the graph Fourier bases, which constitute the orthogonal space. We can apply graph Fourier transform (GFT) to the data at each moment of the signal sources and remove the redundant spatial information\cite{emerging2013}. 

Therefore, we use the time-vertex graph signal (TVGS) processing framework to study the sampling of correlated signals\cite{timevertex,ji2019hilbert}. We model the correlated signal sources as a graph, whose each vertex relates to a continuous function of time. In this way, multiple continuous-time signals are modeled as a continuous time-vertex graph signals (CTVGS). There are some research results about the sampling of time-point diagram signals. Based on joint time-vertex Fourier transform (JFT)\cite{2016JFT}, Yu \emph{et al.}\cite{Yu} proposed a joint sampling scheme, which further reduce the amount of samples of the finite time-vertex graph signal. Ji \emph{et al.}\cite{ji2019hilbert} considered the sampling of TVGS with infinite time components and gave an extension of the separate sampling scheme. However, the above works only consider sampling based on the spectrum of the dimension-reducted signals and thus fail to give the minimum sampling rate required for recovery of the CTVGS.

In this paper, the bandlimited signals are modeled as CTVGS, whose low-dimensional representations are obtained by GFT. When the bandwidths of the signals before and after dimension reduction are given, we prove the lowest sampling rate required for signal recovery and propose a specific method to obtain the minimum sample set. Finally, we verified the feasibility of sampling and recovery through simulation experiments.

\section{Model}
\label{sec:model}

An undirected graph can be represented as $\G = (\V, \E, \A)$, where $\V=\{ v_1, \dots, v_N \}$ is the set of $N$ vertices, $\E$ is the set of edges, and $\A$ is an $N \times N$ symmetric weighted adjacency matrix. The weight $a(i,j)$ of the edge $(v_i, v_j) \in \V$ reflects the degree of the relation of vertex $v_i$ and $v_j$. Adjacency matrix can be decomposed as $ \A =\U \mathbf{\Lambda} \U^H $, where $\U=[\mathbf{u}_1, \dots,  \mathbf{u}_N]$ is an orthonormal matrix consisting of eigenvectors and $\mathbf{\Lambda}$ is a diagonal matrix of eigenvalues $\{ \lambda_i \} _{i=1}^N$ corresponding to $\{ \mathbf{u}_i \}$. Given subsets $\V' \subseteq \V$ and $\Lambda ' \subseteq \Lambda$, $\U_{\V',\Lambda'}$ is constructed by extracting the rows of $\U$ corresponding to $\V'$ and the columns of $\U$ corresponding to $\Lambda '$.


A CTVGS can be written in the following vector form
$$ \X_\V(t) = \left[ \begin{matrix} \X(v_1, t) \\ \dots \\ \X(v_N, t) \end{matrix} \right] \in L^2(\V \times \R), $$
where $\X_v(\cdot) := \X(v,\cdot) \in L^2(\R)$ is a signal on vertex $v$, and $\X(\cdot, t)$ is an ordinary graph signal at instant $t$. The GFT expresses a signal in a linear space based on eigenvectors\cite{emerging2013}, \emph{i.e.},
$$\F_G(\X_\V(t))= \U^H \X_\V(t),$$
whose $i$-th element is 
\begin{equation}
\label{eq:GFT_i}
	\F_G(\X_\V, \lambda_i, t) = u_i^H \X_\V(t) \in L^2(\R).
\end{equation}

Let $\mathscr{B}(x)=\sup\{|\Omega |: |\F_T(x,\Omega)|>0\} \in \overline{\R}_+$ be the bandwidth of $x \in L^2(\R)$, where $\F_T(x,\Omega) = \int_\R x(t)e^{-j\Omega t} dt$ is the Fourier transform(FT) of $x$ and $\Omega$ is the analog angular frequency\cite{shannon}. 
\begin{definition}
    Let $\W_{\B_\V,\B_\Lambda} \subseteq L^2(\V \times \R)$ be the space of CTVGS bandlimited to sets $\B_\V = \{ \mathscr{B}(\X_v), v\in \V \}$ and $\B_\Lambda = \{ \mathscr{B}(\F_G(\X_\V, \lambda, t)), \lambda \in \Lambda \}$.
\end{definition}
For any $\X_\V \in \W_{\B_\V,\B_\Lambda}$, we denote the bandwidth on $\X_v$ as $\B_\V(v) = \mathscr{B}(\X_v)$ and the bandwidth on $\F_G(\X_\V, \lambda, t)$ as $\B_\Lambda(\lambda) = \mathscr{B}(\F_G(\X_\V, \lambda, t))$. Particularly, if $\B_\V(v)$ is a constant value for any $v\in \V$, $\W_{\B_\V,\B_\Lambda}$ is a \emph{signal space with equal bandwidth}. If $\B_\Lambda(\lambda) = 0$ or $\B_\Lambda(\lambda) \geq \max_{v\in \V} \B_\V(v)$ for any $\lambda \in \Lambda$, $\W_{\B_\V,\B_\Lambda}$ is a \emph{signal space with simple bandwidth}.

\section{THE MINIMAL SAMPLE RATE THEOREM}
\label{sec:sampling}

In the general space, the elements in $\B_\V$ is often different from each other. The CTVGS in can be easily partitioned into several signals in the space with equal bandwidth by filtering. Therefore, we mainly introduce the sampling of signals in signal space with equal bandwidth in this section. 

A signal space with equal bandwidth $\W_{\B_\V, \B_\Lambda}$ can be divided into a series of subspaces, and finally be reduced to a signal space with simple bandwidth $\W_{\B_{\V},\B^0_\Lambda}$, that is, $\W_{\B_\V, \B^0_\Lambda} \subseteq \W_{\B_\V, \B^1_\Lambda} \subseteq \dots \subseteq \W_{\B_\V, \B_\Lambda}$. In each division, we record the subspace as $\W_{\B_\V, \B^{i-1}_\Lambda}$, whose complement space (with respect to (w.r.t.) the direct sum) is $\W_{v^i_*}$. $\W_{v^i_*}$ is isomorphic to the space of bandlimited signals in $L^2(\R)$ and can be reconstructed from the samples obtained from the signal on a vertex. The signal in the space with simple bandwidth can be sampled and reconstructed directly. $\W_{\B_\V, \B^{i-1}_\Lambda}$ direct sum $\W_{v^i_*}$ to obtain $\W_{\B_\V, \B^i_\Lambda}$, and recursively, the original signal can be perfectly reconstructed.

\subsection{Space Division}

We now raise how to divide $\W_{\B_{\V},\B^i_\Lambda}$, and the rest are done in the same way. Let 
$$ \lambda^i_* =\mathop{\arg\min}\limits_{\lambda \in \Lambda, \ 0<\B^i_\Lambda(\lambda)< \B_\V} \B^i_\Lambda(\lambda). $$  
According to (\ref{eq:GFT_i}), we have a linear map
\begin{eqnarray*}
	Q_{\lambda^i_*}: \W_{\B_\V, \B^i_\Lambda} & \mapsto & L^2(\R) \\
	\X_\V(v,t) & \mapsto & \F_G(\X, \lambda^i_*, t) = u^H_{\lambda^i_*} \X_\V(t).
\end{eqnarray*}

\begin{lemma}
\label{lem:ker}
    Suppose that $\B^{i-1}_\Lambda$ and $\B^i_\Lambda$ take different values only at $\lambda^i_*$ and $\B^{i-1}_\Lambda(\lambda^i_*) = 0$, $\W_{\B_\V, \B^{i-1}_\Lambda} = {\rm ker}(Q_{\lambda^i_*})$, \emph{i.e.}, $\W_{\B_\V, \B^{i-1}_\Lambda}$ is a subspace of $\W_{\B_\V, \B^i_\Lambda}$.
\end{lemma}

Intuitively, for each $\W_{\B_\V, \B^i_\Lambda}$, we find a $\lambda^i_*$ and get the subspace $\W_{\B_\V, \B^{i-1}_\Lambda}$ by changing $\B^i_\Lambda(\lambda^i_*)$ to 0. After finite steps, we must be able to reduce $\W_{\B_\V, \B_\Lambda}$ to $\W_{\B_\V, \B^0_\Lambda}$ with $\B^0_\Lambda(\lambda) = 0$ or $\B^0_\Lambda(\lambda) \geq \max_{v\in \V} \B_\V(v)$.

Then we consider the complement space of $\W_{\B_\V, \B^{i-1}_\Lambda}$. 

\begin{definition}
    For $\W_{\B_\V, \B^i_\Lambda}$, let $\Lambda^i_0=\{\lambda \in \Lambda : \B^i_\Lambda(\lambda) = 0 \}$. If there is a subset $\V' \subseteq \V$ such that $|\V'|+|\Lambda^i_0| = |\V| = N$ and $\U_{\V'^c,\Lambda^i_0}$ is invertible, $\V'$ is a uniqueness set w.r.t. $\Lambda^i_0$. The collection of uniqueness sets w.r.t. $\Lambda^i_0$ is denoted as $\mathcal{U}(\Lambda^i_0)$. If $\Lambda^i_0 = \emptyset $, let $\mathcal{U}(\Lambda^i_0) = \{ \V \}$.
\end{definition}

For any $\lambda \in \Lambda^i_0$, we have $\B^i_\Lambda(\lambda)=0$ and $\F_G(\X, \lambda, t) = u_{\lambda}^H \X_\V(t) = 0$. Thus
$$ \U_{\V, \Lambda^i_0}^H \X_\V(t) = 0, $$
Let $\V^i_0 \in \mathcal{U}(\Lambda^i_0)$, we obtain
$$ \U_{\V^i_0, \Lambda^i_0}^H \X_{\V^i_0}(t) + \U_{{\V^i_0}^c, \Lambda^i_0}^H \X_{{\V^i_0}^c}(t) = 0, $$
$$ \X_{{\V^i_0}^c}(t) = -(\U_{{\V^i_0}^c,\Lambda^i_0}^H)^{-1} \U_{\V^i_0,\Lambda^i_0}^H \X_{\V^i_0}(t), $$
\begin{equation}
\label{eq:uniq}
	\X_\V(t) = \left[ \begin{matrix} \X_{\V^i_0}(t) \\ \X_{{\V^i_0}^c}(t) \end{matrix} \right] = \left[ \begin{matrix} I \\ -(\U_{{\V^i_0}^c,\Lambda^i_0}^H)^{-1} \U_{\V^i_0,\Lambda^i_0}^H \end{matrix} \right] \X_{\V^i_0}(t) 
\end{equation}

As long as the signals in the uniqueness set are given, $\X_\V$ can be completely determined. That is why we call $\V^i_0$ a uniqueness set. Then $\F_G(\X, \lambda^i_*, t)$ can be written as
\begin{equation}
\label{eq:lambda}
    \F_G(\X, \lambda^i_*, t) = u^H_{\lambda^i_*} \X_\V(t) = E_{\V^i_0} \X_{\V^i_0}(t),
\end{equation}
where $E_{\V^i_0} = \left[ \begin{matrix} I \\ -(\U_{{\V^i_0}^c,\Lambda^i_0}^H)^{-1} \U_{\V^i_0,\Lambda^i_0}^H \end{matrix} \right]$ is a row vector with dimension $|\V^i_0|$. For $v \in \V^i_0$, $E_{\V^i_0}(v)$ is the element corresponding to vertex $v$ in $E_{\V^i_0}$.

In fact, the complement space of ${\rm ker}(Q_{\lambda^i_*})$ is isomorphic to ${\rm im}(Q_{\lambda^i_*})$.

\begin{lemma}
\label{lem:Yi}
    Let $\W_{\B_\V, \B^{i-1}_\Lambda} = {\rm ker}(Q_{\lambda^i_*}) \subseteq \W_{\B_\V, \B^i_\Lambda}$ and $Y_i$ be the complement space of $\W_{\B_\V, \B^{i-1}_\Lambda}$ w.r.t. $\W_{\B_\V, \B^i_\Lambda}$. Then $ Y_i \cong {\rm im}(Q_{\lambda^i_*})$ is the subspace of $L^2(\R)$ with bandwidth $\B^i_\Lambda(\lambda^i_*)$.
\end{lemma}
\begin{proof}
  See Proposition 1 in \cite{ji2020sampling} for details.
\end{proof}

By Lemma \ref{lem:Yi}, we prove that $Y_i$ is isomorphic to ${\rm im}(Q_{\lambda^i_*})$ and determine its bandwidth. The complement space actually used for sampling can then be determined.

\begin{lemma}
\label{lem:im}
    For $v^i_* \in \V^i_0 \in \mathcal{U}(\Lambda^i_0)$ and $E_{\V^i_0}(v^i_*) \ne 0$, let $\W_{v^i_*} = \left\{ \X \in \W_{\B_\V, \B^{i}_\Lambda}: \X_v=0, \forall v \in \V^i_0 \backslash \{v^i_*\} \right\} $ be the complement space of $\W_{\B_\V, \B^{i-1}_\Lambda}$ w.r.t. $\W_{\B_\V, \B^i_\Lambda}$.
\end{lemma}
\begin{proof}
  See \cite{ji2020sampling} for details.
\end{proof}

Intuitively, we can find the subspace $\W_{v^i_*}$ of $\W_{\B_\V, \B^i_\Lambda}$, such that $\W_{v^i_*} = Y_i$. For a CTVGS in $\W_{v^i_*}$, the signals in $\V^i_0$ are zero except the signal on $v^i_*$. Since $\V^i_0$ is a uniqueness set, the signals in ${\V^i_0}^c$ can be determined.

To sum up, in subspace $\W_{\B_\V, \B^{i-1}_\Lambda}$, the sampling and reconstruction of $\X_\V$ can perfectly recover the signals on all vertices in $\V^i_0$ except $v^i_*$, which plus the reconstructed signal in the complement space $\W_{v^i_*}$ to get the signal component in $\W_{\B_\V, \B^i_\Lambda}$. At this time, $\X_{v^i_*}$ is completely determined. \emph{It is worth noting that $\W_{\B_\V, \B^{i-1}_\Lambda}$ and $\W_{v^i_*}$ are not orthogonal here.} Because in $\W_{\B_\V, \B^{i-1}_\Lambda}$, the signal on $v^i_*$ can be determined by the signals on the vertices in $\V^i_0 \backslash \{v^i_* \}$, and it is often not zero. Before reconstructing the signal in the complementary space $\W_{v^i_*}$, the signal component in $\W_{\B_\V, \B^{i-1}_\Lambda}$ should be subtracted first.

\subsection{Sampling in Space with Equal Bandwidth}

It is useful to define a sampling rate for dealing with a subset of $\V \times \R$.

\begin{definition}
    Suppose the sampling set $\S$ is a countable discrete subset of $\V \times \R$. The sampling rate $r(\S)$ is defined as
    $$ r(\S) = \limsup \limits_{t \rightarrow \infty} \frac{|\S \cap (\V \times \left[-t,t  \right])|}{2t}. $$
\end{definition}

Since the original CTVGS $\X$ can be divided into a signal space with simple bandwidth, we first give a sampling theorem for signal component in space with simple bandwidth.

\begin{theorem}
    The minimum sampling rate required for perfect recovery of the signal in signal space with simple bandwidth is
    $$ r^*_0 = 2(N-|\Lambda^0_0|) \B_\V, $$
    where $\B_\V$ is a constant number, because $\B_\V(v)$ is equal for all $v$ here.
\end{theorem}
\begin{proof}
  See Theorem 1 (b) in \cite{ji2020sampling} for details.
\end{proof}

We can recover signal $\X^{i-1}$ in $\W_{\B_\V, \B^{i-1}_\Lambda}$. A signal on $v^i_*$ satisfying Lemma \ref{lem:im} is sampled, and then $\X^{i-1}_{v^i_*}$ is subtracted to recover $\X'_{v^i_*}$. From the definition of uniqueness set, $\X'$ in $\W_{v^i_*}$ can be determined. Thus we have $\X^i = \X^{i-1} + \X'$. That is, we should ensure that we can find the appropriate $v^i_*$ for each space division. To this end, we first give the following definition.

\begin{definition}
\label{def:sequ}
    A sequence $\V_0 \subset \V_1 \subset \cdots \subset \V_k \subseteq \V$ is called an admissible sequence if the following holds:
    \begin{itemize}
        \item [1)]For any $0 \leq i \leq k$, $\V_i$ is the uniqueness set w.r.t. $\Lambda^i_0$.
        \item [2)]$\V_i \backslash \V_{i-1} $ is a singleton $\{ v^i_* \}$.
    \end{itemize}
\end{definition}

When the Definition \ref{def:sequ} 2) is satisfied, the condition of Lemma \ref{lem:im} is naturally true. To prove this conclusion, we need the following lemma.

\begin{lemma}
\label{lem:matrix}
    If the invertible matrix $P$ has a block form:
    $$ P =\begin{bmatrix}
		P' & \alpha \\
		\beta^H & a
	\end{bmatrix}$$
	where $P'$ is a matrix whose dimension is one smaller than $P$, $\alpha$ and $\beta$ are column vectors, $a$ is a complex number. If $P'$ is invertible, then $a - \beta^H {P'}^{-1} \alpha \ne 0$.
\end{lemma}
\begin{proof}
  Perform elementary row transforms on matrix $P$,
  $$\begin{bmatrix}
		P' & \alpha \\
		\beta^H & a
	\end{bmatrix}
	\rightarrow
	\begin{bmatrix}
		P' & \alpha \\
		0 & a-\beta^H{P'}^{-1}\alpha
	\end{bmatrix} $$
	The elementary row transformation does not change the determinant of the matrix, so
	\begin{eqnarray*}
	    0 \ne {\rm det}(P) & = & {\rm det}\left(\begin{bmatrix}
		P' & \alpha \\
		0 & a-\beta^H{P'}^{-1}\alpha \end{bmatrix}  \right) \\
		& = & {\rm det}(P')(a-\beta^H{P'}^{-1}\alpha)
	\end{eqnarray*}
	Since $P'$ is invertible, $a - \beta^H {P'}^{-1} \alpha \ne 0$.
\end{proof}

\begin{proposition}
    Let $\V_{i} = \V_{i-1} \cup \{ v^i_*\} $ and $\Lambda^{i-1}_0=\Lambda^i_0 \cup \{\lambda^i_* \}$, where $\V_{i-1}$ is the uniqueness set w.r.t. $\Lambda^{i-1}_0$, $\V_i$ is the uniqueness set w.r.t. $\Lambda^i_0$. Then we have $E_{\V_i}(v^i_*) \ne 0$.
\end{proposition}
\begin{proof}
  From the definition of uniqueness set, we know that $\U_{\V^c_i, \Lambda^{i}_0}$ and $\U_{\V^c_{i-1}, \Lambda^{i-1}_0}$ are invertible. Suppose that
  $$\U_{\V^c_{i-1}, \Lambda^{i-1}_0} = \begin{bmatrix}
		\U_{\V^c_i, \Lambda^{i}_0} & \alpha \\
		\beta^H & a	\end{bmatrix}, $$
  According to Lemma \ref{lem:matrix}, $a-\beta^H \U^{-1}_{\V^c_i, \Lambda^{i}_0} \alpha \ne 0$. Let 
  $$\U = \begin{bmatrix}
		\U_{\V^c_i, \Lambda^{i}_0} & \alpha & B \\
		\beta^H & a & \phi^H \\
		C & \psi & D \end{bmatrix}$$
  It is the same as (\ref{eq:lambda}) that
  \begin{equation*}
      \begin{aligned}
        & \F_G(\X, \lambda^i_*,t) = u_{\lambda^i_*}^H \X_\V(t) = \begin{pmatrix} \alpha^H & a^H & \psi^H \end{pmatrix} 
		\begin{pmatrix}
			\X_{\V^c_i} \\ \X_{v^i_*} \\ \X_{\V_i}
		\end{pmatrix} \\
		& = 
		\begin{pmatrix}
			\alpha^H & a^H & \psi^H 
		\end{pmatrix} 
		\begin{pmatrix}
			-({\U_{{\V_i}^c,\Lambda^i_0} ^H})^{-1} \U_{\V_i,\Lambda^i_0}^H \\ I
		\end{pmatrix}
		\begin{pmatrix}
			\X_{v^i_*} \\ \X_{\V_i}
		\end{pmatrix}.
      \end{aligned}
  \end{equation*}
	Thus we obtain
	\begin{equation*}
	    \begin{aligned}
	    E_{\V^i_0} & = \begin{pmatrix} \alpha^H & a^H & \psi^H \end{pmatrix} 
	\begin{pmatrix}
		-({\U_{{\V_i}^c,\Lambda^i_0}^H})^{-1} \U_{\V_i, \Lambda^i_0}^H \\ I
	\end{pmatrix}\\
    & = \begin{pmatrix}
		\alpha^H & a^H & \psi^H 
	\end{pmatrix}
	\begin{pmatrix}
		-({\U_{{\V_i}^c, \Lambda^i_0}^H})^{-1}
		\begin{pmatrix}
			\beta & C^H  
		\end{pmatrix}   
		\\I
	\end{pmatrix}
	    \end{aligned}
	\end{equation*}
	where the first element is
	$$ E_{\V_i}(v^i_*) = -\alpha^H ({\U_{{\V_i}^c, \Lambda^i_0}^H})^{-1} \beta + a^H \ne 0. $$
\end{proof}

\begin{proposition}
    The admissible sequence must exist.
\end{proposition}
\begin{proof}
  Since $\U$ is an invertible matrix, its column vectors are linearly independent, and the rank of $\U_{\V,\lambda^0_0}$ is $|\lambda^0_0|$. Additionally, the row rank of a matrix is equal to the rank, so $|\lambda^0_0|$ linearly independent rows of $\U_{\V,\lambda^0_0}$ can be selected, which form the submatrix of $\U_{\V,\lambda^0_0}$ is an invertible square matrix, denoted as $\U_{\V^c_0, \lambda^0_0}$. Now we obtain $\V_0$, which is the uniqueness set w.r.t. $\Lambda^0_0$.
  
  By induction, suppose that $\V_{i-1}, i=1, 2, \dots, k$ is defined. Since $\V_{i-1}$ is the uniqueness set w.r.t. $\Lambda^{i-1}_0$, $\U_{\V^c_{i-1}, \Lambda^{i-1}_0}$ is invertible. Thus, the rank of $\U_{\V^c_{i-1}, \Lambda^{i}_0}$ is $|\lambda^{i}_0|$ and the submatrix $\U_{\V^c_i, \lambda^i_0}$ is invertible. Then $\V_i$ is defined.
\end{proof}

\begin{theorem}
\label{th:minr}
    For a signal space with equal bandwidth $\W_{\B_\V, \B_\Lambda}$, let $\V_0 \subset \V_1 \subset \V_2 \subset \cdots \subset \V_k$ be an admissible sequence. Then the minimum sampling rate required to perfectly recover $\X \in \W_{\B_\V, \B_\Lambda}$ is 
    \begin{equation*}
		r(\V, \Lambda) = r^*_0 + 2\sum^k_{i=1} \B_{\Lambda}(\lambda^i_*).
	\end{equation*}
\end{theorem}
\begin{proof}
  See Theorem 2 in \cite{ji2020sampling} for details.
\end{proof}

\subsection{Sampling of general CTVGS}

For general CTVGS, the bandwidths at the vertices are often different, so we can first filter the original signal, get the part with the same bandwidth and the remaining part. The part with the same bandwidth is in the signal space with equal bandwidth, which can be sampled by the above method. The remaining part also can be divided into signals with the same bandwidth and the remaining one by frequency shift and filtering. Repeat the operations to divide the original CTVGS into multiple signals in the space with equal bandwidth. Finally, add up the signals to get the original CTVGS.

Suppose that the output signal of each filter is $\X_j \in \W_{\B_{\V'_j}, \B_{\Lambda'_j}}$. According to Theorem \ref{th:minr}, the minimum sampling rate required for the perfect reconstruction of $\X_j$ is $r(\V'_j, \Lambda'_j)$. So we have the following theorem.
\begin{theorem}
    For space $\W_{\B_\V, \B_\Lambda}$, let subspaces after filtering for $J$ times be the space with equal bandwidth, and $\V'_j$ and $\Lambda'_j$ are vertex set and graph frequency set of the $j$-th subspace, respectively. The minimum sampling rate required for perfectly recovering the signal in $\W_{\B_\V, \B_\Lambda}$ is
    $$ r = \sum^J_{j=1} r(\V'_j, \Lambda'_j). $$
\end{theorem}

\section{Simulation}
\label{sec:exp}

Since it is hard to simulate continuous-time signals directly, in this section, we use a discrete data set as the raw data for simulation. Our proposed sampling scheme focuses on determining the sampling rate based on the spectrum of CTVGS. When a sequence is recorded at a sufficient rate, whose spectrum is extended periodically. The spectrum of the sequence within one period is the same as that of the continuous-time signal. Therefore, we construct a TVGS with a discrete data set for simulation and are able to illustrate the feasibility of our proposed sampling scheme.

We simulate 4 channels of real electroencephalogram (EEG) signals, and each channel is considered a vertex. Every 1024 consecutive timesteps of the EEG data were recorded as one signal. Ultimately, we simulate TVGS with a size of $N = 4$ vertices, and a discrete sequence of 1024 timesteps relates to each vertex.

To construct the graph construction, we calculate the covariance matrix of the de-meaned TVGS, which is the adjacency matrix $\A$\cite{GFTW}. The signal is a short time, so we consider its adjacency matrix not varying with time. We then perform an eigenvalue decomposition of $\A$ to obtain the graph Fourier bases.

We pass 100 signals through a low-pass filter to obtain TVGS in space with equal bandwidth, \emph{i.e.}, $\X \in \W_{\B_\V, \B_\Lambda}$, and the filter bandwidth is $\B_\V$. By adjusting the filter, we performed 10 experiments with different bandwidth. The reconstructed TVGS is noted as $\hat{\X}$.

To measure the difference between signals $\X_a$ and $\X_b$, we define the normalized Root Mean Square Error (NRMSE)
$$ \text{NRMSE}(\X_a,\X_b) = \frac{ \| \X_a-\X_b \|_2 }{ \| \X_a \|_2 }. $$
The results obtained by our proposed scheme are compared and analyzed with those obtained by the existing separation sampling scheme. 

    \begin{figure}[htbp]
    \centering    
    \subfigure[]
    {
	    \includegraphics[scale=0.29]{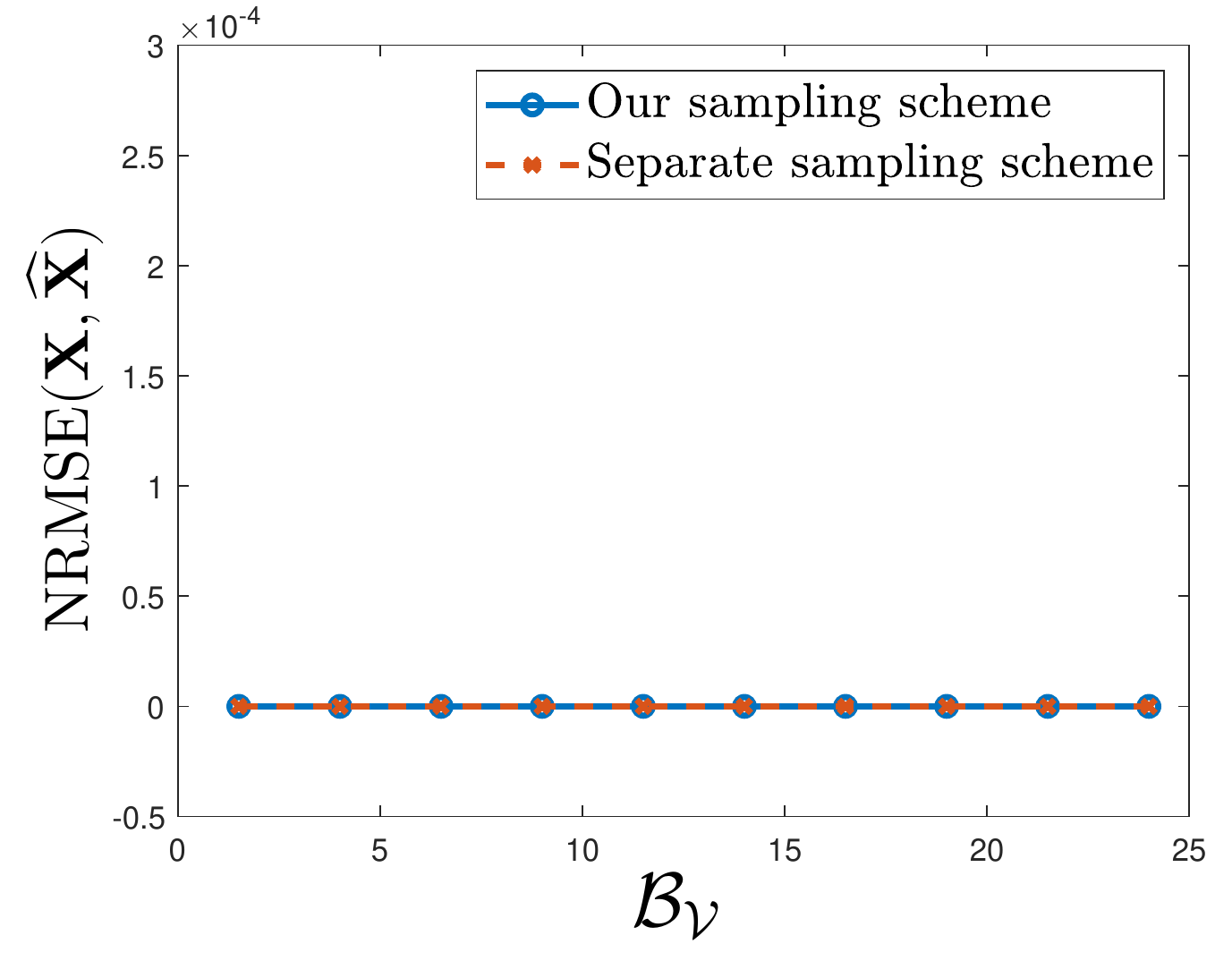}   
    }
    \subfigure[]
    {
	    \includegraphics[scale=0.29]{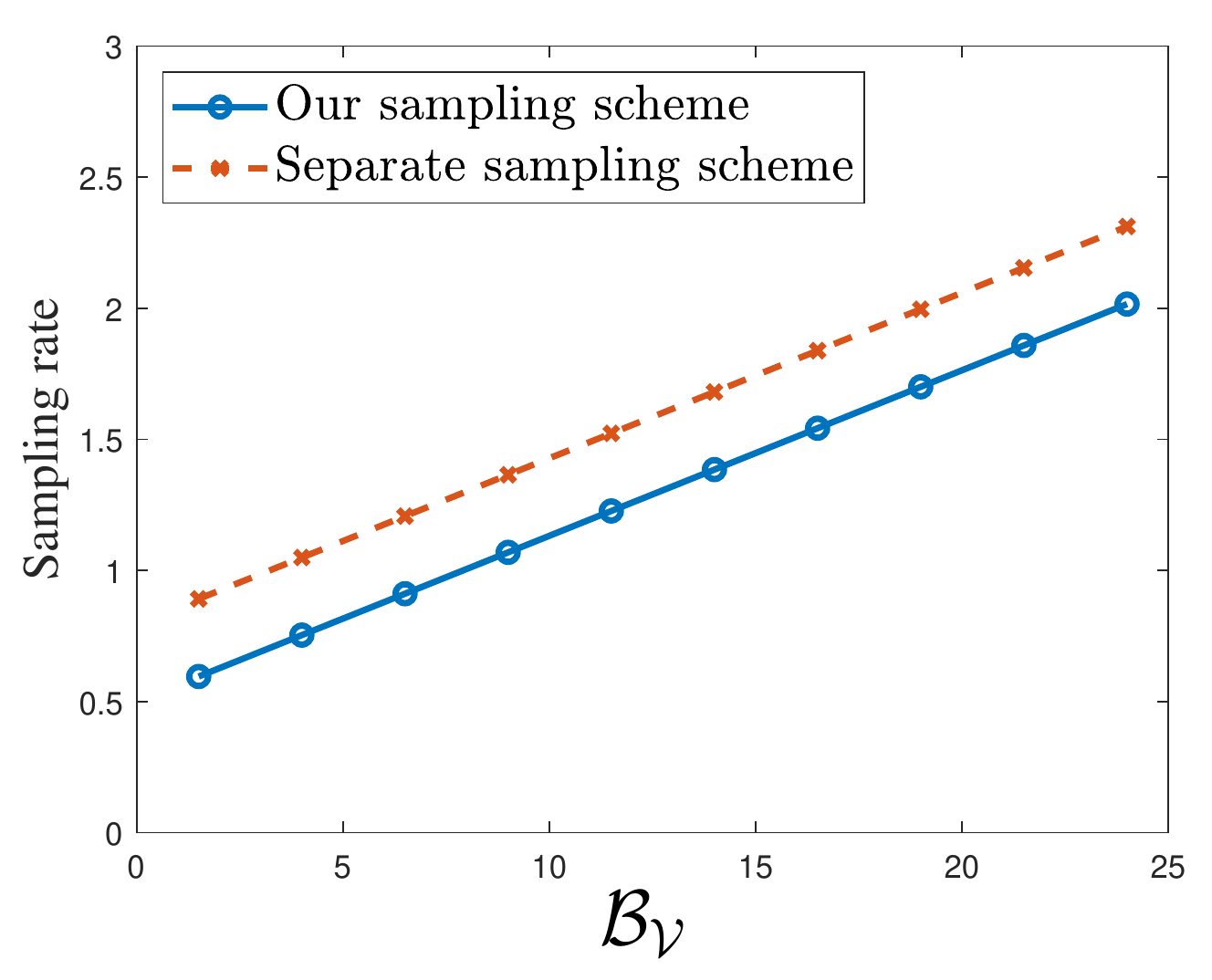}   
    }
    \caption{Comparison of the mean value of the results of 100 experiments. (a) is the variation curves of NRMSE with different $\B_\V$. (b) is the variation curves of sampling rate with different $\B_\V$.}
    \label{fig:exp} 
    \end{figure}

Both our method and the separate sampling scheme can recover the bandlimited signals perfectly, shown in Fig \ref{fig:exp} (a). And regardless of the bandwidth, our method can sample and reconstruct at a lower sampling rate, shown in Fig \ref{fig:exp} (b).

\section{Conclusion}
\label{sec:conclusion}

In this paper, we combine PCA and the sampling of multiple bandlimited signals by the time-vertex graph signal processing framework. Based on the idea of space division, we prove the sampling theory of CTVGS. In addition, the method of sampling at the lowest sampling rate is proposed. In the future, we will consider issues such as the confidence level of dimension reduction when the graph construction change over time.

\bibliographystyle{IEEEbib}
\bibliography{paper}

\end{document}